\documentclass[a4paper,11pt]{amsart}
\usepackage[T1]{fontenc}
\usepackage[utf8]{inputenc}
\usepackage{lmodern}
\usepackage{amsthm}

\usepackage{amssymb,amsmath}
\usepackage{latexsym}
\usepackage[usenames,dvipsnames]{color}
\usepackage[all]{xy}
\usepackage{graphicx}
\usepackage{mathrsfs}
\usepackage{stmaryrd}

\usepackage{paralist}
\usepackage{stackrel}

\newenvironment{lemmalist}{\begin{compactenum}[\itshape i.)]}{\end{compactenum}}
\DeclareMathOperator{\image} {\mathrm{im}}
\newcommand{\Cinfty}{\mathscr{C}^{\infty}}

\newcommand{\st}{^*}

\newcommand{\g}{{\mathfrak{g}}{}}

\newcommand{\hc}{\mathsf{h}}
\newcommand{\LL}{\mathscr{L}}
\newcommand{\bbN}{\mathbb{N}}

\newcommand{\norm}{{\|}}
\DeclareMathOperator{\dd}{d}
\newcommand{\Hom}            {\operatorname{\mathsf{Hom}}}
\newcommand{\Diff}            {\operatorname{\mathsf{Diff}}}
\newcommand{\Vect}            {\operatorname{\mathsf{Vect}}}
\newcommand{\Poly}            {\operatorname{\mathsf{Poly}}}

\DeclareMathOperator{\id}    {\mathsf{id}}
\newcommand\ot[2]{\stackrel{#1}{#2}}

\newtheorem{theorem*}{Theorem}
\newtheorem{theorem}{Theorem}[section]
\newtheorem{lemma}[theorem]{Lemma}
\newtheorem{metatheorem}[theorem]{Metatheorem}
\newtheorem{corollary}[theorem]{Corollary}
\newtheorem{proposition}[theorem]{Proposition}
\theoremstyle{definition}
\newtheorem{definition}[theorem]{Definition}
\theoremstyle{remark}
\newtheorem{remark}[theorem]{Remark}
\newtheorem{example}[theorem]{Example}
\numberwithin{equation}{section}

\title{ Rigidity of infinitesimal momentum maps}
\author{Chiara Esposito}
\address{Faculty of Mathematics and Computer Science,
Department of Mathematics,
Chair of Mathematics X (Mathematical Physics), Julius Maximilian University of W\"{u}rzburg, Germany }
\email{chiara.esposito@mathematik.uni-wuerzburg.de}
\author{Eva Miranda}
\address{Departament de Matem\`{a}tica Aplicada I,
 Universitat Polit\`{e}cnica de Catalunya and BGSMath, Barcelona, Spain }
\email{eva.miranda@upc.edu}
\thanks{Eva Miranda is partially supported by Ministerio de Econom\'{i}a y Competitividad project GEOMETRIA ALGEBRAICA, SIMPLECTICA, ARITMETICA Y APLICACIONES with reference: MTM2012-38122-C03-01/FEDER and by the European Science Foundation network CAST}
\begin{document}

\maketitle
\date{\today}

\begin{abstract}
  In this paper we prove rigidity theorems for Poisson Lie group
  actions on Poisson manifolds. In particular, we prove
  that close infinitesimal momentum maps associated to
  Poisson Lie group actions are equivalent using a normal form theorem for SCI spaces.
  When the Poisson structure of the acted manifold is integrable,
  this yields rigidity also for lifted actions to the symplectic
  groupoid.
\end{abstract}

\section{Introduction}
\label{sec:intro}

In 1961 Palais proved that close actions of compact Lie groups on
compact manifolds can be conjugated by a diffeomorphism \cite{Palais1961}.
The interest of this rigidity theorem relies on the approximation of actions
by nearby ones. As application of this rigidity theorem of Palais we
can recover normal form theorems such as the linearization theorem by
Bochner \cite{Bochner1945}. Several generalizations of this result
have been obtained in \cite{Ginzburg1992} and \cite{Miranda2007}
for the case of symplectic structures and in \cite{Miranda2012} for
the case of Hamiltonian actions of semisimple Lie algebras on
Poisson manifolds.

In this paper we generalize a rigidity result from \cite{Miranda2012}
to the context of Poisson Lie groups and pre-Hamiltonian actions. This addresses a question considered by Ginzburg in \cite{Ginzburg1996}.

Poisson Lie group actions
on Poisson manifolds with non-trivial Poisson structures
appear naturally in the study of $R$
matrices. As explained in \cite{Kosmann-Schwarzbach2004} this generalization
of Hamiltonian actions is useful to take into account the properties of the
dressing transformations under the \emph{hidden symmetry group}
in the case of $R$-matrices. For these,
the notion of momentum mapping for Poisson manifolds coincides
with the monodromy matrix of the associated linear system.
Thus, rigidity for Poisson Lie group actions can be useful
to understand the stability of the integrable systems associated
to $R$-matrices. On the other hand, it is worthwhile to mention that
the generalization
to Poisson Lie group actions represents a first step towards
quantization of symmetries, as studied in \cite{Esposito2012}.
The first main result of this paper is the rigidity of Poisson Hamiltonian
actions, i.e. Poisson Lie actions generated by a G-equivariant
momentum map $\mu:M\longrightarrow G^*$. More explicitely,
\begin{theorem*}
  Let $G$ be a compact semisimple Poisson Lie group $G$ acting
  on a compact Poisson manifold $(M,\pi)$ in a Hamiltonian
  fashion given by the momentum map $J_0 : M\longrightarrow G^\ast$.

  There exist a positive integer $l$ and two positive real numbers
  $\alpha$ and $\beta$ (with $\beta<1<\alpha$) such that, if $J_1$ is
  another momentum map on $M$ with respect to the same Poisson
  structure and Poisson Lie group, satisfying
  \begin{equation}
      \| J_0 - J_1 \|_{2l-1}
      \leq
      \alpha
      \quad {\mbox { and }} \quad
      \| J_0 - J_1 \|_{l}
      \leq
      \beta
  \end{equation}
  then, there exists a Poisson diffeomorphism $\psi$ of class
  $C^k$, for all $k\geq l$, on $M$ such that
  $J_1\circ \psi =J_0$.
\end{theorem*}
In other words, close Poisson Hamiltonian actions are equivalent.
The proof uses  a global linearization theorem due to Ginzburg and
Weinstein \cite{Ginzburg1992} and the rigidity result for
Hamiltonian actions on Poisson manifolds obtained in \cite{Miranda2012}.
As a consequence of this rigidity
theorem for momentum maps in the Poisson Lie group setting,
we obtain rigidity for lifted actions to the symplectic groupoid
when the Poisson structure is integrable.

As pointed out in \cite{Ginzburg1996}, in many cases
it is enough to consider the infinitesimal version of the momentum map.
Poisson actions generated by an infinitesimal momentum map
are what we call pre-Hamiltonian actions.
The advantage of considering the infinitesimal momentum map
relies in the fact that existence and uniqueness are much simpler to prove
and we have a big class of examples given by semisimple Lie algebras.
Infinitesimal momentum maps are
the local counterpart to momentum maps and topology on the
acted manifold is an obstruction to its integration to global
momentum maps. This is also the case when the Poisson structure
on $G$ is not trivial but there are additional obstructions as
shown in \cite{Esposito2012a}.

The main result of this paper is the rigidity of the infinitesimal
momentum map.
\begin{theorem*}
  Let us consider an pre-Hamiltonian action of a semisimple compact Poisson Lie group $(G, \pi_G)$
  on a compact Poisson manifold $(M,\pi)$ with infinitesimal momentum map
  $\alpha$.

  There exist a positive integer $l$ and two positive real numbers
  $a$ and $b$ (with $b<1<a$) such that, if $\tilde{\alpha}$ is
  another infinitesimal momentum map on $M$ with respect to the same Poisson
  structure, satisfying
  \begin{equation}
      \| \alpha - \tilde{\alpha} \|_{2l-1}
      \leq
      a
      \qquad
      \| \alpha - \tilde{\alpha} \|_{l} \leq b
  \end{equation}
  then, there exists a map $\Phi: \Omega^1(M)\to \Omega^1(M)$
  of class $\mathscr{C}^{k}$, for all   $k\geq l$
  preserving the Lie algebra structure of $\Omega^1(M)$
  and the differential $\dd$, such that
  $\Phi(\alpha_X) = \tilde\alpha_X$.
\end{theorem*}
The proof uses techniques native to geometrical analysis and
an abstract normal form theorem from \cite{Miranda2012}.
This abstract normal form encapsules a Newton's iterative method
used by Moser and Nash to prove the inverse function theorem in
infinite dimensions (see for example \cite{Hamilton1982}).
Newton's method is used to prove normal form results by
approximating the solution by means of an iterative process.
The solution is then presented as a limit. The abstract normal
form for SCI spaces presented in \cite{Miranda2012} allows to
prove normal forms results (and in particular, linearization
and rigidity theorems) without having to plunge into the
details of the iterative method. In this paper we provide a
new application of this normal form for SCI spaces.
The abstract normal form theorem in \cite{Miranda2012} has had
other applications in the theory of generalized complex manifolds
(see \cite{Bailey2013} and \cite{Bailey2014}) and a variant of it
to normal forms in a neighbourhood of a symplectic leaf  of a
Poisson manifold \cite{Marcut2014}. In this paper we provide a new
application of this normal form for SCI spaces
As in \cite{Miranda2012} we first prove an infinitesimal
rigidity result and then we apply the SCI normal form
theorem to conclude rigidity. Our theorem can be
seen as another reincarnation of Mather's principle
\emph{\lq\lq infinitesimal stability implies stability\rq\rq}
(see \cite{Mather1968} and its sequel).


\section{Preliminaries: Poisson Lie groups and pre-Hamiltonians actions}
\label{sec:mm}

In this section we introduce a generalization of the notion of Hamiltonian actions
in the setting of Poisson Lie groups acting on Poisson manifolds.
Let us recall that a Poisson Lie group is a pair $(G,\pi_G)$, where $G$ is a Lie group
and $\pi_G$ is a multiplicative Poisson structure. The  Lie algebra $\mathfrak{g}$ corresponding
to the Lie group $G$ is equipped with the 1-cocycle,
\begin{equation}
    \delta
    =
    d_e\pi_G : \mathfrak{g}\rightarrow  \mathfrak{g}\wedge \mathfrak{g},
\end{equation}
which defines a Lie algebra structure on the dual vector space $\mathfrak{g}^*$.
For this reason, the pair $(\mathfrak{g} , \delta)$ is said to be a Lie bialgebra.
If $G$ is connected and simply connected there is a one-to-one correspondence
between the Poisson Lie group $(G, \pi_G)$ and the Lie bialgebra $(\mathfrak{g}, \delta)$,
as proven in \cite{Drinfeld1983} (for this reason we assume this
hypothesis to hold throughout this paper). The dual Poisson Lie group $G^*$ is
defined to be the Lie group associated to the Lie algebra $\mathfrak{g}\st$.
Given a Poisson Lie group $(G,\pi_G)$ and a Poisson manifold $(M,\pi)$,
we introduce the following
\begin{definition}
  \label{def:PoissonAction}
  The action of $(G,\pi_G)$ on $(M,\pi)$ is called \textbf{Poisson action}
  if the map $\Phi:G\times M\rightarrow M$ is Poisson, that is
  \begin{equation}
      \lbrace f \circ \Phi, g \circ \Phi \rbrace_{G \times M}
      =
      \lbrace f,g \rbrace_M \circ \Phi \qquad \forall f,g\in \Cinfty(M)
  \end{equation}
  where the Poisson structure on $G\times M$ is given by $\pi_G\oplus\pi$.
\end{definition}
Observe that if $G$ carries the zero Poisson structure $\pi_G=0$, the action is Poisson
if and only if it preserves $\pi$.
In general, when $\pi_G\neq 0$, the structure $\pi$ is not invariant with respect to the action.

Let $\Phi: G \times M \to M$ be a Poisson action and $\widehat{X}$ the fundamental vector field
associated to any element $X \in \mathfrak{g}$. For each $X \in \mathfrak{g}$ we can also define
the left invariant 1-form $\theta_X$ on $G^*$ with value $X$ at $e$.
\begin{definition}[Lu, \cite{Lu1990}, \cite{Lu1991}]
  \label{def:MomentumMap}
  A \textbf{momentum map} for the Poisson action $\Phi: G\times M\rightarrow M$
  is a map $J: M\rightarrow G^*$ such that
  \begin{equation}
    \label{eq:MomentumMap}
      \widehat{X}
      =
      \pi^{\sharp}(J^*(\theta_{X}))
  \end{equation}
\end{definition}
We recall that a momentum map is said to be $G$-equivariant if and only if it
is a Poisson map, i.e.
\begin{equation}
    J_*\pi
    =
    \pi_{G^*}.
\end{equation}
Finally, we can say that a \textbf{Poisson Hamiltonian action} in this context
is a Poisson action induced by an equivariant momentum map.
This definition generalizes Hamiltonian actions in the canonical setting
of Lie groups acting on Poisson manifolds. Indeed, we notice that, if
the Poisson structure on $G$ is trivial, the dual $G^*$ corresponds
to the dual of the Lie algebra $\mathfrak{g}^*$, the one-form $J^*(\theta_{X})$
is exact and the infinitesimal generator $\widehat{X}$ is a Hamiltonian
vector field.

The notion of Poisson Hamiltonian action can be further generalized
by using a weaker definition of momentum map, first introduced by Ginzburg
in \cite{Ginzburg1996} and further developed in \cite{Esposito2012}.
The basic idea is to consider the infinitesimal version of an equivariant
momentum map generating the fundamental vector fields of a Poisson action.
For this purpose, it is useful to recall that
given a Poisson structure $\pi$, the anchor map $\pi^{\sharp}$ defined as
$\pi^{\sharp}(\alpha) := \pi(\alpha, \cdot)$, defines a skew-symmetric
operation $[\cdot,\cdot]_{\pi} : \Omega^1(M) \times \Omega^1(M) \to \Omega^1(M)$.
This operation is given by the general formula
\begin{equation}
  \label{eq:LieForms}
    [\alpha,\beta]_{\pi}
    =
    \LL_{\pi^{\sharp}(\alpha)}\beta - \LL_{\pi^{\sharp}(\beta)}\alpha - d(\pi(\alpha,\beta))
\end{equation}
Furthermore, it provides $\Omega^1(M)$ with a Lie algebra structure such that
$\pi^{\sharp}: T^*M \to TM$ is a Lie algebra homomorphism (see Theorem 4.1 in\cite{Vaisman1994}).
\begin{definition}[Ginzburg, \cite{Ginzburg1996}]
  \label{def:Pre}
  Let $(M,\pi)$ be a Poisson manifold and $(G,\pi_G)$ a Poisson Lie group.
  \begin{itemize}
    \item[(i)] An \textbf{infinitesimal momentum map} is a map
    $\alpha:\mathfrak{g} \to \Omega^{1} (M): X \mapsto \alpha_X$ such that
    it generates the action by
	  \begin{equation}
	      \widehat{X}
	      =
	      \pi^{\sharp}(\alpha_X)
	  \end{equation}
    and satisfies the conditions to be a Lie algebra homomorphism
    \begin{equation}
      \label{eq:LieAlgMorph}
	      \alpha_{[X,Y]}
	      =
	      [\alpha_{X},\alpha_Y ]_{\pi}
    \end{equation}
    and a cochain map
    \begin{equation}
      \label{eq:Chain}
	      \dd\alpha_{X}
	      =
	      \alpha \wedge \alpha \circ \delta(X)
    \end{equation}
    \item[(ii)] A \textbf{pre-Hamiltonian action} is a Poisson action of $(G,\pi_G)$ on
    $(M,\pi)$ induced by an infinitesimal momentum map $\alpha: \mathfrak{g} \to \Omega^{1} (M)$.
  \end{itemize}
\end{definition}
Clearly this notion is weaker than the Hamiltonian notion, as it does not reduce
to the canonical one when the Poisson structure on $G$
is trivial. In fact, if $\pi_G = 0$ we have $\delta = 0$ and Eq.~(\ref{eq:Chain})
implies that $\alpha_X$ is a closed form, but in general this form is not exact unless
$M$ is simply connected.
If $M$ is not simply connected we can get examples in the symplectic realm like
rotations on a torus or more sophisticated ones for general Poisson structures.
\begin{example}
  \label{thm:toricbtorus}
  Consider the torus $\mathbb T^2$, with Poisson structure
  $\pi= {\sin\theta_1}\frac{\partial}{\partial \theta_1}\wedge \frac{\partial}{\partial \theta_2}$
  where the coordinates on the torus are $\theta_1,\theta_2\in\left[0,2\pi\right]$ .
  This Poisson structure is symplectic away from the set $Z=\{\theta_1 \in \{0, \pi\}\}$ and the Poisson
  structures satisfies a transversality condition at the vanishing set. This Poisson structure
  pertains to a class called $b$-Poisson structures (or $b$-symplectic structures)
  studied in \cite{Guillemin2014}. The circle action of rotation on the $\theta_2$
  coordinate defines a pre-Hamiltonian action on $\mathbb T^2$.
  Indeed it is possible to associate a $b$-symplectic form to this Poisson structure
  (see \cite{Guillemin2014}) and work with  $b$-symplectic actions. In this case
  $\frac{1}{\sin\theta_1}d\theta_1\wedge d\theta_2$. The circle action of rotation on
  the $\theta_2$ coordinate is pre-Hamiltonian and the associated one-form is
  $\frac{1}{\sin\theta_1} d\theta_1$ (see \cite{Guillemin2015} for properties
  of these actions on $b$-Poisson manifolds).
\end{example}
Furthermore, it is clear that any Poisson Hamiltonian action is pre-Hamiltonian because
the equivariant momentum map $J$ induces the infinitesimal one $\alpha$
by $\alpha_X = J^*(\theta_X)$. But not every infinitesimal momentum map arises from a
momentum map $J$. Here we recall an example of pre-Hamiltonian action which is not
Poisson Hamiltonian (see Remark 3.3 in\cite{Ginzburg1996}).

The study of the conditions in which an infinitesimal momentum map
determines a momentum map can be found in \cite{Esposito2012a}. The authors here
proved that if $M$ and $G$ are simply connected and $G$ is compact, then
${\mathscr D} = \{ \alpha_{\xi}-\theta_{\xi},\ \xi\in {\mathfrak g}\}$
generates an involutive distribution on $M \times G^*$
and a leaf $\boldsymbol{\mu}_{\mathscr{F}}$ of $\mathscr D$
is a graph of a momentum map if
\begin{equation}
    \pi (\alpha_{\xi} ,\alpha_{\eta}) - \pi_{G^*}(\theta_{\xi},\theta_{\eta})\vert_{\mathscr F}
    =
    0 ,
    \qquad
    \xi ,\eta \in {\mathfrak g}.
\end{equation}
The advantage of working with infinitesimal
momentum map is that the study of its existence and uniqueness is much
simpler than for the $G$-equivariant momentum map of Definition \ref{def:MomentumMap}.
In particular, it has been proven in \cite{Ginzburg1996} that any action of a
compact group with $H^2(\mathfrak{g}) = 0$ admits an infinitesimal momentum map.
Since assuming $\g$ semisimple implies automatically $H^2(\mathfrak{g}) = 0$, we can
conclude that whenever $\g$ is semisimple the action admits an infinitesimal momentum map.


\section{Rigidity of Poisson Hamiltonian actions}
\label{sec:RigidityHamiltonian}

As introduced in Section \ref{sec:intro}, by \textit{rigidity} of the actions
we mean that close actions are equivalent, i.e. they are conjugated
by a diffeomorphism. In what follows, we recall the rigidity theorem
in the context of Hamiltonian actions and we prove the rigidity of
Poisson Hamiltonian actions as its direct consequence.

\subsection{The case of zero Poisson structure on the Poisson Lie group $G$}
\label{sub:Zero}

Let us consider a Lie group $G$ acting on a Poisson manifold $M$ and
assume that the action is Hamiltonian, i.e. there exists a $G$-equivariant momentum map
$\mu : M \to \g^*$. The rigidity of Hamiltonian actions has been showed
in \cite{Miranda2012} by proving a rigidity theorem for the momentum map.
More precisely,
\begin{theorem}[Miranda, Monnier, Zung  \cite{Miranda2012}]
  \label{thm:mainglobal}
  Consider a compact Poisson manifold $(M,\pi)$ and a Hamiltonian action on $M$
  given by the momentum map $\lambda : M\longrightarrow \g^\ast$ where
  $\g$ is a semisimple Lie algebra of compact type.

  There exist a positive integer $l$ and two positive real numbers
  $\alpha$ and $\beta$ (with $\beta<1<\alpha$) such that, if $\mu$ is
  another momentum map on $M$ with respect to the same Poisson
  structure and Lie algebra, satisfying
  \begin{equation}
      \| \lambda - \mu \|_{2l-1} \leq \alpha
      \quad {\mbox { and }} \quad \|
      \lambda - \mu \|_{l} \leq \beta
  \end{equation}
  then, there exists a diffeomorphism $\psi$ of class
  $\mathscr{C}^{k}$, for all $k\geq l$, on $M$ such that
  $\mu\circ \psi =\lambda$.
\end{theorem}
The main idea to prove this theorem is to approximate a given momentum map by an
iteration of momentum maps.
As explained in \cite{Guillemin2002}, a first approach to prove the equivalence of
Lie group actions on manifolds would following steps.
In general a Lie group action gives an element in $\mathscr{M} = \Hom(G,\Diff(M))$
and we can consider the additional action,
$$
    \begin{array}{rcl}
        \beta :& \Diff(M)  \times \mathscr{M}&\longmapsto \mathscr{M} \\
        & (\phi,\alpha) & \mapsto \phi\circ\alpha\circ \phi^{-1}
    \end{array}
$$
Two actions $\alpha_0$ and $\alpha_1$ are conjugate if they are on
the same orbit under  $\beta$  so, in particular, if $\beta$ has open orbits
the action is rigid. Observe that,
\begin{itemize}
  \item The tangent space to the orbit of $\beta$ coincides with  $1$-coboundaries
  of the group cohomology with coefficients in $V=\Vect(M)$ and
  the tangent space to $\mathscr{M}$ are the $1$-cocycles.
  \item The  generalized Whitehead lemma implies that for compact $G$ the
  cohomology group $H^{1}(G;\Vect(M))$ vanishes. This phenomenon is
  known as infinitesimal rigidity.
  In this case the tangent space to the orbit equals the tangent space to $\mathscr{M}$.
  \item If $\mathscr{M}$ is a manifold (or tame Fréchet)  we can apply the inverse function
  theorem  Nash-Moser to go from the tangent space to the manifold. We can use this fact
  to prove that $\beta$ has open manifolds and thus the action is rigid.
\end{itemize}
In general it is hard to verify the  \lq\lq  tame Fréchet'' condition but we can apply the
method used in the proof of Nash-Moser's theorem (Newton's iterative method).
This methods allows to proof several results of type
\emph{infinitesimal rigidity implies rigidity}. For Hamiltonian actions on Poisson
manifold the authors \cite{Miranda2012} consider the Chevalley-Eilenberg complex associated to the
representation given by the momentum map following the next steps:
\begin{enumerate}
  \item Assume  that the close momentum maps are $\mu: M \to \g^*$ and
  $\lambda:M\to \g^*$. The difference $\phi = \mu-\lambda$ defines a
  1-cochain of the complex which is a near 1-cocycle.
  \item We define $\Phi$ as the time-1-map of the Hamiltonian vector field
  $X_{S_{t}(h(\phi))}$ with  $h$ the homotopy operator and $S_t$ is a smoothing operator.
  \item The Newton iteration is given by,
  $$
      \phi_d
      =
      \phi^{1}_{X_{S_{t}(h(\eta_d))}}
  $$
  with $\eta_d = (\mu-\lambda)\circ \phi_{d-1} $. This converges to a Poisson
  diffeomorphism that conjugates both actions.
\end{enumerate}
Convergence is a \emph{hard} part of the proof. In order to circumvent these
difficulties, in \cite{Miranda2012} a strong use of geometric analysis tools is
performed to check this using the paraphernalia of SCI spaces. In particular the
theorem needed to prove convergence is the abstract normal form presented in
Section \ref{sub:NormalSCI}

\subsection{The case of non-trivial Poisson structure on the Poisson Lie group $G$}
\label{sec:rha}

In this section we prove the rigidity of Poisson Hamiltonian actions,
i.e. we show that two close Poisson Hamiltonian
actions with momentum maps $J_0, J_1: M\longrightarrow G^*$
are equivalent. This result has been obtained by combining well-known results of
Ginzburg and Weinstein concerning linearization of Poisson-Lie groups with
the ridigity theorem \ref{thm:mainglobal} for canonical momentum maps.

Observe that since a Poisson structure on a Poisson Lie group (with Poisson structure $\pi_G$)
must vanish at $e\in G$, its linearization at $e$ is well-defined
(recall that $d_e\pi_G:\mathfrak g\longrightarrow \mathfrak g\wedge \mathfrak g$).
The following theorem says that if $G$ is compact and semisimple, the Poisson
structure $\pi_G$ is linearizable, thus equivalent to $d_e\pi_G$ by diffeomorphisms.
\begin{theorem}[Ginzburg, Weinstein \cite{Ginzburg1992}]
  \label{thm:ginzburg}
    Let $G$ be a compact semisimple Poisson Lie group then the
    dual Poisson Lie group $G^*$ is globally diffeomorphic to
    $\g^*$ with the linear Poisson structure defined as
    $\{f,g\}_\eta=<\eta, [df_\eta, dg_{\eta}]>$.
\end{theorem}
Combining Theorem \ref{thm:ginzburg} and Theorem \ref{thm:mainglobal}
to obtain rigidity for the Poisson Hamiltonian action, as stated below
\begin{theorem}
  \label{thm:mainhamiltonian}
  Let $G$ be a compact semisimple Poisson Lie group $G$ acting on a
  compact Poisson manifold $(M,\pi)$ in a Hamiltonian fashion given
  by the momentum map $J_0 : M\longrightarrow G^\ast$.

  There exist a positive integer $l$ and two positive real numbers
  $\alpha$ and $\beta$ (with $\beta<1<\alpha$) such that, if $J_1$ is
  another momentum map on $M$ with respect to the same Poisson
  structure and Poisson Lie group, satisfying
  \begin{equation}
      \| J_0 - J_1 \|_{2l-1}
      \leq
      \alpha
      \quad {\mbox { and }} \quad
      \| J_0 - J_1 \|_{l}
      \leq
      \beta
  \end{equation}
  then, there exists a Poisson diffeomorphism $\psi$ of class
  $C^k$, for all $k\geq l$, on $M$ such that
  $J_1\circ \psi =J_0$.
\end{theorem}

\begin{proof}
Denote by $\Phi$ the linearizing Poisson diffeomorphism\footnote{The
differentiability class can be assumed to be $k$ by the construction
in the proof of Theorem \ref{thm:ginzburg}} $\Phi: G^*\longrightarrow \mathfrak g^*$
given by theorem \ref{thm:ginzburg} and consider the compositions
$\mu_0=\Phi\circ J_0$ and $\mu_1=\Phi\circ J_1$.
The mappings  $\mu_0: M\longrightarrow \mathfrak{g}^*$ and
$\mu_1: M\longrightarrow \mathfrak{g}^*$ are canonical momentum maps
and we may consider the infinitesimal Hamiltonian actions of
$\mathfrak g$ ($\beta_0$ and $\beta_1$). These actions integrate to
infinitesimal standard Hamiltonian actions of the Lie group $G$ which
preserve the Poisson structure on $M$. We may now apply Theorem
\ref{thm:mainglobal} to obtain a diffeomorphism $\psi$ such that
$\mu_1\circ \psi =\mu_0$ and therefore $J_1\circ \psi =J_0$.
\end{proof}
This result is just telling us that rigidity of the standard momentum map
implies rigidity of Lu's momentum map.
In general this will work whenever we have a linearization theorem
for the Poisson Lie group $G$.
\begin{metatheorem}
  \label{thm:mainhamiltonian}
  Whenever the Poisson Lie group structure in $G$ is
  linearizable, the rigidity of the momentum map
  $\mu : M \to \g^*$ implies the rigidity of the momentum
  map $J: M \to G^*$ for linearizable Poisson Lie groups $G$.
\end{metatheorem}
The linearization of Poisson Lie groups has been studied by
Enriquez, Etingof and Marshall \cite{Enriquez2005} in the context of
quasitriangular Poisson Lie groups and further generalized to
coboundary Poisson Lie groups by Alekseev and Meinrenken \cite{Alekseev2013}.
In particular, for coboundary Poisson Lie groups the authors
define a modified exponential $\mathsf{E}: \mathfrak g^* \to G^*$ and prove
that it is a Poisson diffeomorphism. If a rigidity result would work for coboundary
Poisson Lie groups then the metatheorem above would imply rigidity
for this class too because $J_i = \mathsf{E}\circ\mu_i$.
To the authors' knowledge, such a rigidity result for
Hamiltonian actions is not known to hold in general for non-semisimple Lie groups.
\begin{remark}
  It would be possible to relax the SCI-hypotheses in order to prove
  rigidity for Poisson Lie group actions on compact manifolds.
  The SCI-apparatus is indeed thought for the semilocal case
  (neighbourhood of a compact invariant submanifold).
  However, thanks to the SCI-scheme the rigidity statement for
  compact manifolds is automatically valid in the semilocal setting
  (due to the need to control the convergence of the radii of
  shrank neighbourhoods in the iterative process).
\end{remark}

\subsubsection{An application to groupoids}

Theorem \ref{thm:mainhamiltonian} has a direct application to the study of
momentum maps lifted to symplectic groupoids.
Let us consider an integrable Poisson manifold $M$ and
its symplectic groupoid $\Sigma(M) \rightrightarrows M$.
We recall that, as proved in \cite{Xu2003}, if one has a Poisson Hamiltonian
action of $(G, \pi_G)$ on $(M,\pi)$ with momentum map $J: M\to G^*$, then
$J^{\Sigma(M)}: \Sigma(M) \to  G^*$ is exact:
\begin{equation}
    J^{\Sigma(M)}(x)
    =
    J(t(x))J(s(x))^{-1},
\end{equation}
where $s,t$ are the source and target maps.
Thus, using the fact that Poisson morphisms can be integrated (see \cite{Crainic2004a})
we get the following,
\begin{corollary}
    Given two close  momentum maps
    $J_i:M\longrightarrow G^*$, $i= 1, 2$ on an integrable Poisson
    manifold $M$, then there exists a symplectic groupoid morphism
    $\phi$ on $\Sigma(M)$ such that the corresponding lifted moment maps
    $J_i^{\Sigma(M)}$ satisfy $J_1^{\Sigma(M)} = J_2^{\Sigma(M)}\circ \phi$.
\end{corollary}
In other words, rigidity of the momentum maps implies rigidity of the corresponding
lifted momentum maps.
The general case of momentum maps on symplectic groupoids is still open and this corollary gives a
motivating example to investigate on the rigidity of $J^{\Sigma(M)}$ when $J$ does not exist.

\section{Rigidity of pre-Hamiltonian actions}
\label{sec:Main}

In this section we prove the main result of this paper, the rigidity
of pre-Hamiltonian actions. More precisely, we consider two Poisson
actions generated by the infinitesimal momentum maps $\alpha$ and
$\tilde{\alpha}$ and we prove a rigidity property:
\textit{close implies equivalent}. The proof follows the same lines
discussed in Section \ref{sub:Zero}.

First, we have to set up the concepts of \textit{close} and
\textit{equivalent} for infinitesimal momentum maps.
We can define the topology by using the associated infinitesimal
momentum maps, i.e. we can also use the $\mathscr{C}^{k}$-norm
of the infinitesimal momentum map $\alpha: \mathfrak g\to \Omega^{1}(M)$
and work with $\alpha_X$, for $X\in \mathfrak g$  as
mappings $\alpha_X: M\to  T^*M$.
On the other hand, two infinitesimal momentum maps are said to be
\textit{equivalent} if there exists a morphism of Lie algebras
conjugating them.

As in the Hamiltonian setting, we aim to prove that infinitesimal rigidity
implies rigidity, thus the first step is to consider the Chevalley-Eilenberg
complex associated to the infinitesimal momentum map.
The first cohomology group of this complex can be interpreted as
infinitesimal deformations and when it vanishes we obtain infinitesimal rigidity.
Then, using the techniques of SCI-spaces we can prove the equivalence of
infinitesimal momentum maps via Lie algebra morphisms. More explicitely,
let $\alpha$ and $\tilde{\alpha}$ be two close infinitesimal
momentum maps. The idea is to construct a sequence $\alpha_n$ which
are equivalent, with $\alpha_0 = \alpha$ and such that $\alpha_n$
tends to $\tilde{\alpha}$ when $n$ tends to $+\infty$.
\begin{enumerate}
    \item We consider the difference $\beta = \alpha - \tilde{\alpha}$,
    which defines a 1-cochain of the associated complex which is a near 1-cocycle.
    \item We define $\Phi$ as the time-1-map of the vector field
    $X_{\hc(\beta)} = \pi^\sharp(\hc(\beta))$ with $\hc$ the homotopy operator.
    \item The Newton iteration is given by,
    $$
        \Phi_n = \phi^{1}_{X_{(\hc(\beta_n))}}
    $$
    with $\beta_n = \Phi_{n-1} \circ(\alpha - \tilde{\alpha})$. This converges to a
    Lie algebra morphism that conjugates both momentum maps.
  \end{enumerate}
Instead of checking convergence of this sequence of we are going to
use a normal form theorem for SCI-spaces.


\subsection{A Chevalley-Eilenberg complex and infinitesimal rigidity}
\label{sec:ce}

First we define the Chevalley-Eilenberg complex associated to an
infinitesimal momentum map $\alpha: \g \to \Omega^1(M)$
and discuss the properties that will be used to prove
the rigidity theorem. The infinitesimal momentum map
defines a representation of the Lie algebra $\g$ on the space
of 1-forms on $M$ as we prove in the following
\begin{lemma}
  Let $\alpha: \g \to \Omega^1(M): X \mapsto \alpha_{X}$ be the infinitesimal
  momentum map. It defines a representation $\rho$ of $\g$ on $\Omega^1(M)$
  by
  \begin{equation}
    \label{eq: rep}
	    \rho_{X}(\beta)
	    :=
	    [\alpha_{X}, \beta]_{\pi} 
  \end{equation}
  for any $X\in \g $.
\end{lemma}

\begin{proof}
  This is a direct consequence of properties of the Lie bracket
  $[\cdot,\cdot]_{\pi}$ and of $\alpha$ since we have:
  \begin{align*}
    \rho_X \rho_Y(\beta)- \rho_Y \rho_X (\beta)
    &=
    [\alpha_{X}, [\alpha_{Y}, \beta]_{\pi}]_{\pi}-[\alpha_{Y}, [\alpha_{X}, \beta]_{\pi}]_{\pi} \\
    &=
    [[\alpha_{X}, \alpha_{Y}]_{\pi}, \beta]_{\pi} \\
    &=
    [\alpha_{[X,Y]}, \beta]_{\pi} \\
    &=
    \rho_{[X,Y]}(\beta).
  \end{align*}
\end{proof}
Thus, for $q\in\mathbb N$, $C^q(\g, \Omega^1(M)) = {\Hom}(\bigwedge^q{\frak{g}},\Omega^1(M))$
is the space of alternating $q$-linear maps from $\g$ to
$\Omega^1(M)$, with the convention $C^0(\g,\Omega^1(M)) = \Omega^1(M)$.
The associated differential is denoted by $\partial_i$.
Explicitly, we have
\[
    \xymatrix{
    \Omega^1(M)
    \ar[r]^-{\partial_0} &%
    C^1(\g, \Omega^1(M))
    \ar[r]^-{\partial_1} &%
    C^2(\g, \Omega^1(M))
    }
\]
where
\begin{equation}
\label{eq:DifferentialCE}
\begin{split}
    \partial_0(\beta)(X)
    &=
    \rho_{X}(\beta)
    = [\alpha_{X}, \beta]_{\pi},  \\
    \partial_1(\gamma)(X \wedge Y)
    &=
    \rho_{X}(\gamma(Y))-\rho_{Y}(\gamma(X))-\gamma([X,Y]),
\end{split}
\end{equation}
for any $\beta\in\Omega^1(M)$, $\gamma\in C^1(\g, \Omega^1(M))$ and
$X,Y \in\frak g$.
These differentials satisfy $\partial_i\circ \partial_{i-1}=0$ and
we can define the quotients
$$
    H^{i}(\frak g, \Omega^1(M))
    =
    \ker(\partial_i)/\image(\partial_{i-1}) \quad
  \forall i \in \mathbb{N}.$$
The first cohomology group can be interpreted as infinitesimal deformations of
the infinitesimal momentum maps modulo trivial deformations.
In the compact semisimple case, it is known that the first and second cohomology
groups vanish, so we have the infinitesimal rigidity.
To prove that infinitesimal rigidity implies rigidity we need to prove
that our spaces comply with the SCI-spaces requirement (for more details about SCI-spaces
see \cite{Miranda2012}). In particular, certain inequalities have to be
checked for the homotopy operators, necessary to control the
loss of differentiability in the iterative process.
In the Hamiltonian case, the trick used
in \cite{Miranda2012} and \cite{Conn1985} in order to prove the desired inequalities
is to first use Sobolev metrics and then Sobolev inequalities and then take the
real part. For the Chevalley-Eilenberg complex defined above
we need those inequalities applied to mappings
$\alpha: \mathfrak g\to \Omega^{1}(M)$
and work with  $\alpha_X$, for $X\in \mathfrak g$
as mappings $\alpha_X : M\to  T^*M$.
Since $M$ is compact, Sobolev inequalities holds too. A different way to do
this is to consider Sobolev norms in the space of one-forms\footnote{For
one-forms on oriented manifolds, we may consider the higher degree versions of
the following norm:  $<\alpha, \beta> = \int_X \alpha\wedge *\beta $
where $*\beta$ stands for the Hodge dual of $\beta$.} and
$\mathscr{C}^{k}$-topology for the space of one-forms (see for instance \cite{Dodziuk1981} or
\cite{Goldstein2006}) and adapt the same steps.
\begin{lemma}
  \label{lem:Homotopy}
  In the Chevalley-Eilenberg complex associated to $\rho$:
  \[
      \xymatrix{
      \Omega^1(M)
      \ar[r]^-{\partial_0} &%
      C^1(\g, \Omega^1(M))
      \ar[r]^-{\partial_1} &%
      C^2(\g, \Omega^1(M))}
  \]
  there exists a chain of homotopy operators
  \[
  \xymatrix{
  \Omega^1(M)\ar[r]^-{\partial_0} &%
  C^1(\frak g, \Omega^1(M))\ar[r]^-{\partial_1}\ar@<1ex>[l]^-{\hc_0}  &%
  C^2(\frak g, \Omega^1(M))\ar@<1ex>[l]^-{\hc_1} }.
  \]
  such that
  \begin{align*}
      \partial_0\circ \hc_0 + \hc_{1}\circ\partial_{1}
      &=
      \id_{C^{1}(\frak g, \Omega^1(M)))}\\
      \partial_1\circ \hc_1 + \hc_{2}\circ\partial_{2}
      &=
      \id_{C^{1}(\frak g, \Omega^1(M)))}.
  \end{align*}
  Moreover,  for each $k$, there exists a real constant $C_k>0$ such that
  \begin{equation}
  \label{5.7}
    \norm{\hc_j(S)}\norm_{k, r}\leq C_k\norm{S}\norm_{k+s, r}, \quad j=0,1,2
  \end{equation}
  for all $S\in C^{j+1}(\frak g, \Omega^1(M))$
\end{lemma}

\begin{proof}
  We apply the same strategy of \cite{Miranda2012} replacing the Sobolev inequalities
  for smooth function by the analogous for differential forms. A key
  point is that those Sobolev norms are invariant by the
  action of the Lie group which is linear.
  The linearity of the action is needed to
  decompose the Hilbert space into spaces which are invariant.

  In our case we can assume that this action is also linear using
  an appropriate $G$-equivariant embedding by virtue of Mostow-Palais theorem
  (\cite{Mostow1957}, \cite{Palais1961})\footnote{Using an orthonormal basis in
  the vector space $E$ for this
  action we can define the corresponding Sobolev norms in the ambient
  spaces provided by the Mostow-Palais embedding theorem. This norm is
  invariant by the action of $G$ (we can even assume $G$ is a subgroup of the orthogonal group).}.
  As it was done in \cite{Miranda2012}, we can check the regularity properties of the homotopy
  operators  with respect to these Sobolev norms and then deduce, as a consequence,
  regularity properties of the initial norms by looking at the real
  part. The proof holds step by step by replacing the standard Sobolev
  inequalities by the ones for differential one-forms.
\end{proof}

\begin{remark}
  If we restrict to exact forms, it follows immediately that
  \begin{equation}
      \partial_{0}(\beta)(X)
      =
      [\alpha_{X},\beta]_{\pi}
      =
      [\dd H_X, \dd f]_{\pi}
      =
      \dd\lbrace H_X, f\rbrace.
  \end{equation}
  Thus, the Chevalley-Eilenberg complex $C^q(\g,
  \Omega^1(M))$ recovers $C^q(\g,C^\infty(M))$.
\end{remark}


\subsection{An abstract normal form for SCI-spaces}
\label{sub:NormalSCI}

  As announced, we need to recall the normal form theorem proved in \cite{Miranda2012}
  for SCI-spaces. SCI-spaces (where SCI stands for \emph{scaled $C^\infty$-type}) are
  a generalization of scaled
  spaces and tame Fréchet spaces. This analytical apparatus is needed to prove normal form theorems
  in the most possible general setting which includes neighbourhood of a point, a compact invariant
  submanifold or a compact manifold. We refer to \cite{Miranda2012} for the basic definitions of
  SCI-spaces, SCI-groups  and SCI-actions. It is good to keep in mind the following archetypical
  example: an example of SCI-spaces is the set of Poisson structures, an example of SCI-group is
  the group of diffeomorphism (which can be germified, semilocal or global), and in this case
  an example of SCI-action is the push-forward of a Poisson structure via a diffeomorphism.

  The scheme of proof of a normal form theorem in this abstract setting is the following:
  \begin{enumerate}
    \item ${\mathcal G}$ (for instance diffeomorphisms) which acts on a set ${\mathcal S}$
    (of structures).
    \item We consider the subset of structures in normal forms $\mathcal N$ inside ${\mathcal S}$.
    \item The equivalence of an element in ${\mathcal S}$ to a normal form is understood in
    the following way: for each element $f \in {\mathcal S}$ there is an element
    $\phi \in {\mathcal G}$ such that $\phi\cdot f \in {\mathcal N}$.
  \end{enumerate}
  For practical purposes it is convenient to assume that a ${\mathcal S}$
  (in the example above, the set of Poisson structures) is a subset of a linear space
  $\mathcal T$ (in the example above $\mathcal T$ would be the set of bivector fields).
  The SCI-group $\mathcal G$ acts on $\mathcal T$ and the set of normal forms
  $\mathcal{N} = \mathcal{F}\cap\mathcal{S}$ where $\mathcal F$ is a linear
  subspace of $\mathcal T$.

  The following theorem is an abstract normal form theorem for SCI-spaces.
  In order to apply
  it to particular situations, we need to identify the sets ${\mathcal S}$,
  ${\mathcal F}$, ${\mathcal T}$ and the SCI-group $\mathcal G$ in each case.
  We also need to identify $\mathcal G_0$ a closed subgroup of $\mathcal G$ which is not
  necessarily an SCI-subgroup. As a consequence the equivalence to the normal form is
  given by the existence of $\psi\in\mathcal G$ (or in a closed subgroup)
  $\mathcal G_0$ for each $f\in\mathcal{S}$ such that $\psi\cdot f\in \mathcal{N}$.

  \begin{theorem}[Miranda, Monnier, Zung \cite{Miranda2012}]
    \label{thm:NormalForm}
    Let $\mathcal{T}$ be a SCI-space, $\mathcal{F}$ a SCI-subspace of
    $\mathcal{T}$, and $\mathcal{S}$ a subset of $\mathcal{T}$. Denote
    $\mathcal{N} = \mathcal{F}\cap\mathcal{S}$. Assume that there is a
    projection $\pi: \mathcal{T}\longrightarrow \mathcal{F}$ (compatible
    with restriction and inclusion maps) such that for every $f$ in
    $\mathcal{T}_{k,\rho}$, the element $\zeta(f)=f-\pi(f)$ satisfies
    \begin{equation}
      \label{eqn:proj}
        \|\zeta(f)\|_{k,\rho}
        \leq
        \|f\|_{k,\rho} \Poly(\|f\|_{[(k+1)/2],\rho})
    \end{equation}
    for all $k \in \bbN$ (or at least for all $k$ sufficiently large),
    where $[\cdot]$ is the integer part.

    Let $\mathcal{G}$ be an SCI-group acting on $\mathcal{T}$ by a
    linear left SCI-action and let $\mathcal{G}^0$ be a closed subgroup of
    $\mathcal{G}$ formed by elements preserving $\mathcal{S}$.
    Let $\mathcal{H}$ be a SCI-space and assume that there exist maps
    $\mathsf{H} : \mathcal{S} \longrightarrow \mathcal{H}$ and $\Phi:
    \mathcal{H} \longrightarrow \mathcal{G}^0$ and an integer $s\in\bbN$
    such that for every $0 < \rho \leq 1$, every $f$ in $\mathcal{S}$
    and $g$ in $\mathcal{H}$, and for all $k$ in $\bbN$ (or at least for
    all $k$ sufficiently large) we have the three properties:

    \begin{equation}
    \label{eqn:estimate-H}
      \begin{split}
      \| \mathsf{H}(f) \|_{k,\rho} &\leq \|\zeta(f)\|_{k+s,\rho}
    \Poly(\|f\|_{[(k+1)/2]+s,\rho})   \\
     &+ \|f\|_{k+s,\rho}\|\zeta(f)\|_{[(k+1)/2]+s,\rho}
    \Poly(\|f\|_{[(k+1)/2]+s,\rho}) \ ,
      \end{split}
    \end{equation}

    \begin{equation}
      \label{eqn:estimate-Exp}
        \| \Phi(g) - \id \|_{k,\rho'}
        \leq
        \| g \|_{k+s,\rho}\Poly(\|g\|_{[(k+1)/2]+s,\rho})
    \end{equation}
    and
    \begin{equation}
    \label{eqn:estimate-Exp-bis}
      \begin{split}
        \| \Phi(g_1)\cdot f - \Phi(g_2)\cdot f \|_{k,\rho'}
        & \leq  \| g_1 - g_2 \|_{k+s,\rho} \| f \|_{k+s,\rho}
        \Poly(\| g_1 \|_{k+s,\rho}, \| g_2 \|_{k+s,\rho})  \\
        &  \quad +  \| f \|_{k+s,\rho} \Poly_{(2)}(\| g_1 \|_{k+s,\rho}, \| g_2
        \|_{k+s,\rho})
      \end{split}
    \end{equation}
    if $\rho'\leq \rho(1-c\|g\|_{2,\rho})$ in (\ref{eqn:estimate-Exp})
    and $\rho'\leq \rho(1-c\|g_1\|_{2,\rho})$ and $\rho'\leq
    \rho(1-c\|g_2\|_{2,\rho})$ in (\ref{eqn:estimate-Exp-bis}).

    Finally, for every $f$ in $\mathcal{S}$ denote $\phi_f = \id + \chi_f =
    \Phi\big( \mathsf{H}(f) \big) \in\mathcal{G}^0$ and assume that
    there is a positive real number $\delta$ such that we have the
    inequality
    \begin{equation}
    \label{eqn:estimate-zeta}
        \|\zeta(\phi_f \,.\,f) \|_{k,\rho'}
        \leq
        \|\zeta(f)\|_{k+s,\rho}^{1+\delta}
        Q(\|f\|_{k+s,\rho},\|\chi_f\|_{k+s,\rho}, \|\zeta(f)\|_{k+s,\rho},\|f\|_{k,\rho})
    \end{equation}
  (if $\rho'\leq \rho(1-c\|\chi_f\|_{1,\rho})$) where $Q$
  is a polynomial of four variables and whose degree in the first
  variable does not depend on $k$
  and with positive coefficients.
  Then there exist $l\in\bbN$ and two positive constants $\alpha$ and
  $\beta$ with the following property: for all $p \in \bbN \cup
  \{\infty\}, p \geq l$, and for all $f\in\mathcal{S}_{2p-1,R}$ with
  $\|f\|_{2l-1,R}<\alpha$ and $\|\zeta(f)\|_{l,R}<\beta$, there exists
  $\psi\in\mathcal{G}^0_{p,R/2}$ such that $\psi\cdot f\in
  \mathcal{N}_{p,R/2}$.
\end{theorem}
  Here we use the following notation:
  \begin{itemize}
  \item $\Poly(\| f \|_{k,r})$ stands for a polynomial term in $\| f \|_{k,r}$ where the
  polynomial has positive coefficients and does not depend on $f$ (though it may depend on
  $k$ and on $r$ continuously).
  \item The notation $\Poly_{(p)}(\| f \|_{k,r})$, where $p$ is a strictly positive integer,
  denotes a polynomial term in
  $\| f \|_{k,r}$ where the polynomial has positive coefficients and
  does not depend on $f$ (though it may depend on $k$ and on $r$
  continuously) and {\it which contains terms of degree greater or
  equal to} $p$.
  \end{itemize}


\subsection{Rigidity of infinitesimal momentum maps}
\label{sub:rigpre}

  Finally, we can state the main theorem of this paper in which
  we prove the rigidity of pre-Hamiltonian actions of Poisson Lie groups
  on Poisson manifolds with infinitesimal momentum map
  $\alpha: \g\rightarrow \Omega^1(M)$. In Section \ref{sec:ce} we introduced
  the associated Chevalley-Eilenberg complex and the infinitesimal rigidity.
  In the following we use Theorem \ref{thm:NormalForm} to prove the equivalence
  of two close momentum maps.
  In order to prove that our spaces satisfy the hypotheses of Theorem \ref{thm:NormalForm}
  we need some technical lemmas (they are generalizations of the Lemmas
  of Appendix 2 in \cite{Miranda2012}).
  \begin{lemma}
  \label{lem:Estimate1}
  Let $r > 0$ and $0< \eta<1$ be two positive numbers.
  Consider a one-form $\omega$ on a ball $B_{r(1+\eta)}\in \mathbb R^n$
  and a smooth map $\chi: B_r \to \mathbb R^n$ such that $\chi(0) = 0$
  and $\|\chi\|_{1,r}<\eta$. Then the
  composition $(id + \chi^*)\circ\omega$ is a one-form on a ball $B_r$ which
  satisfies the following inequality:
  \begin{equation}
  \label{eq:Estimate1}
      \|(id+\chi^*)\circ\omega\|_{k,r}
      \leq
      \|\omega\|_{k,r(1+\eta)}(1+P_k(\|\chi\|_{k,r}))
  \end{equation}
  where $P_k$ is a polynomial of degree $k$ with vanishing constant
  term (and which is independent of $\omega$ and $\chi$).
\end{lemma}
\begin{proof}
  The proof follows the same line of Lemma B.1 in \cite{Miranda2012}
\end{proof}
Before stating the second technical lemma we need to recall some
basic results of Poisson calculus, following \cite{Bhaskara1988}.
In particular, we introduce a Lie derivative $\LL_\alpha$
in the direction of a 1-form $\alpha$. We will see that $\LL_\alpha$
integrate to a flow $\Phi^*_t$ on $\Omega^1(M)$ which preserves the
Lie algebra structure.
Recall from Section \ref{sec:mm} that the space of one-forms on $M$
is endowed with a Lie bracket $[\cdot , \cdot]_\pi$ defined by
Eq.~(\ref{eq:LieForms}). We set
\begin{equation}
  \label{eq:LieDerivativeForms}
    \LL_\alpha{\beta}
    =
    [\alpha , \beta]_\pi
\end{equation}
It is clear that this operation makes $\Omega^1(M)$ into
a $\Omega^1(M)$-module, sice it is a Lie algebra.
Now, let us consider a vector field $\pi^\sharp(\alpha)$
and assume that its flow $\phi_t$ is defined for all $t \in \mathbb{R}$.
\begin{theorem}[Ginzburg, \cite{Ginzburg1996}]
  \label{thm:PoissonFlow}
  There exist families of fiber-wise linear automorphisms $\Phi^*_t$
  of the vector bundle $p: T^*M \to M$ covering $\phi_{-t}$ by
  $p\circ \Phi^*_t = \phi_{-t}$ such that
  \begin{lemmalist}
    \item $\Phi^*_t$ is a flow:
    $$
        \Phi^*_{t_1 + t_2}
        =
        \Phi^*_{t_1}\Phi^*_{t_2}
    $$
    for any $t_1, t_2\in \mathbb{R}$.
    \item For any $\beta\in \Omega^1(M)$ the time-dependent
    form $\beta(t) = \Phi^*_t \beta$ is a unique solution of
    the differential equation
    \begin{equation}
      \label{eq:PoissonFlow}
        \frac{\dd \beta(t)}{\dd t}
        =
        \LL_\alpha \beta(t),
        \quad
        \beta(0)
        =
        \beta
    \end{equation}
  \end{lemmalist}
\end{theorem}
From now on we say that the flow $\Phi^*_t$ is a \textbf{Poisson flow}
and it is known that it has many interesting features; here we just
recall the one that will be necessary for our purpose.
\begin{proposition}[Ginzburg, \cite{Ginzburg1996}]
  \label{prop:FlowProp}
  The Poisson flow $\Phi^*_t$ has the following properties:
  \begin{lemmalist}
    \item It preserves the algebra structure:
    \begin{equation}
      \label{eq:FlowAlgebra}
        \Phi^*_t(\beta_1 \wedge \beta_2)
        =
        \Phi^*_t \beta_1 \wedge \Phi^*_t \beta_2
    \end{equation}
    \item $\Phi^*_t : \Omega^1(M) \to \Omega^1(M)$ is a Lie
    algebra morphism:
    \begin{equation}
      \label{eq:FlowLieMorph}
        \Phi^*_t[\beta_1 , \beta_2]_\pi
        =
        [\Phi^*_t\beta_1 , \Phi^*_t\beta_2]_\pi
    \end{equation}
  \end{lemmalist}
\end{proposition}
The definition and the properties of Poisson flows allow us to
prove the following
\begin{lemma}
\label{lem:Estimate2}
  Let $r>0$ and $0<\eta<1$ be two positive numbers.
  With the notations above, we have the two following properties:
  \begin{lemmalist}
    \item For any positive integer $k$ we have
    \begin{equation}
        \| \partial (\alpha-\tilde{\alpha}) \|_{k,r}
        \leq
        c \| \alpha-\tilde{\alpha}\|_{k+1,r}^2 \,,
    \end{equation}
    where $c$ is a positive constant independent of
    $\alpha$ and $\tilde{\alpha}$.
    \item There exists a constant $a > 0$ such that if
    $\| \alpha-\tilde{\alpha} \|_{s+2,r(1+\eta)}< a \eta$, 
    then we have, for any positive integer $k$:
    \begin{equation}
        \| \Phi^*\circ\alpha  - \tilde{\alpha}\|_{k,r}
        \leq
        \| \alpha-\tilde{\alpha} \|_{k+s+2,r(1+\eta)}^2 P(\| \alpha-\tilde{\alpha} \|_{k+s+1,r(1+\eta)})
    \end{equation}
    where $P$ is a polynomial with positive coefficients, independent of
    $\alpha$ and $\tilde{\alpha}$.
  \end{lemmalist}
\end{lemma}

\begin{proof}
  \begin{lemmalist}
  \item Let us consider a basis $\{X_1,\hdots,X_n\}$ of the Lie algebra $\g$
  and the structure constants of the Lie algebra $c_{ij}^k$ defined by
  $$
      [X_i, X_j]
      =
      \sum_{k=1}^n c_{ij}^k X_k.
  $$
  Let us denote by $\alpha_i$ the one-form $\alpha_{X_i}$ associated to the
  element $X_i\in \g$ and by $\beta$ the difference $\alpha-\tilde{\alpha}$ of
  two infinitesimal momentum maps (with respect to the same Poisson structure).
  Using the definition of the Chevalley-Eilenberg differential $\partial$ introduced
  in Section \ref{sec:ce}, we have:
  \begin{equation}
      \partial \beta (X_i\wedge X_j)
      =
      [ \alpha_{i}, \beta_{j} ]_{\pi} - [ \alpha_{j}, \beta_{i} ]_{\pi}-\beta([X_i, X_j]),
  \end{equation}
  where $[\cdot, \cdot]_\pi$ is the Lie bracket induced on $\Omega^1(M)$
  by the Poisson structure $\pi$ on $M$. This allows us to write the following equality:
  \begin{equation}
      [ \beta_{i}, \beta_{j} ]_{\pi}
      =
      [ \alpha_{i}, \alpha_{j} ]_{\pi} - [ \alpha_{i},  \tilde{\alpha}_{j} ]_{\pi}
      -
      [ \tilde{\alpha}_{i}, \alpha_{j} ]_{\pi} + [ \tilde{\alpha}_{i}, \tilde{\alpha}_{j} ]_{\pi}
  \end{equation}
  Since $\alpha$ and $\tilde{\alpha}$ are infinitesimal momentum maps, we can use
  Eq.~(\ref{eq:LieAlgMorph}) and we get
  \begin{equation}
      [\alpha_{i}, \alpha_{j} ]_{\pi}
      =
      \sum_{k=1}^n c_{ij}^k\alpha_k
      \quad\text{and}\quad
      [\tilde{\alpha}_{i}, \tilde{\alpha}_{j} ]_{\pi}
      =
      \sum_{k=1}^n c_{ij}^k \tilde{\alpha}_k
  \end{equation}
  Therefore, we obtain:
  \begin{equation}
      \partial \beta (X_i\wedge X_j)
      =
      [ \beta_{i}, \beta_{j} ]_{\pi} \,.
  \end{equation}
Finally, we just write the following estimates :
\begin{equation}
\| \partial \beta \|_{k,r}  \leq  n(n-1) \| \pi \|_{k,r} \| \beta
\|_{k+1,r}^2 .
\end{equation}

  \item The difference $\beta = \alpha-\tilde{\alpha}$ can be seen
  as an 1-cochain in the Chevalley-Eilenberg complex
  $C^\bullet (\g,\Omega^1(M))$. Thus, $\hc(\beta)$ is an element
  of $\Omega^1(M)$ that we can contract with the Poisson structure $\pi$
  to get a vector field $\pi^\sharp(\hc(\beta))$. Let $\Phi_t^*$ be
  the associated Poisson flow introduced in Theorem
  \ref{thm:PoissonFlow} and consider
  \begin{equation}
  \label{eq:QuadError}
      \Phi^*\alpha_i - \tilde{\alpha}_i
      =
      \Phi^*\alpha_i - \Phi^*\tilde{\alpha}_i +  \Phi^*\tilde{\alpha}_i- \tilde{\alpha}_i .
  \end{equation}
  Using the definition of Poisson flow given by Eq.~(\ref{eq:PoissonFlow}) we have,
  for each $i\in\{1,\hdots,n\}$,
  \begin{equation}
  \label{eq: tec1}
  \begin{split}
      \Phi^*\tilde{\alpha}_i- \tilde{\alpha}_i
      &\ot{(\ref{eq:PoissonFlow})}{=}
      \int_0^1 \Phi^*_t \LL_{\hc(\beta)} \tilde{\alpha}_i \dd t\\
      &\ot{(\ref{eq:LieDerivativeForms})}{=}
      \int_0^1 \Phi^*_t \left[\hc(\beta), \tilde{\alpha}_i \right] \dd t\\
      &\ot{(\ref{eq:DifferentialCE})}{=}
      -\int_0^1 \Phi^*_t \partial (\hc(\beta))_i \dd t.
  \end{split}
  \end{equation}
  From Lemma \ref{lem:Homotopy}, we have the equality
  \begin{equation}
  \label{eq:Hom}
		\beta
		=
		\partial \hc (\beta) + \hc\partial(\beta),
  \end{equation}
  Then, substituting Eq.~(\ref{eq:Hom}) in Eq.~(\ref{eq: tec1}) we have
  \begin{equation}
  \label{eq: tec2}
      \Phi_t^*(\tilde{\alpha}_i)- \tilde{\alpha}_i
      =
      -
      \int_0^1 \Phi_t^*\beta_i \dd t
      +
      \int_0^1 \Phi_t^*\hc\partial(\beta)_i \dd t
  \end{equation}
  Thus
  \begin{equation}
    \begin{split}
      \Phi^*(\alpha_i)- \tilde{\alpha}_i
      & \ot{(\ref{eq:QuadError})}{=}
      \Phi^*(\alpha_i) - \Phi^*(\tilde{\alpha}_i) + \Phi^*(\tilde{\alpha}_i)- \tilde{\alpha}_i\\
	    & \ot{(\ref{eq: tec2})}{=}
	    \Phi^*\beta_i - \int_0^1 \Phi_t^*\beta_i \dd t + \int_0^1 \Phi_t^*\hc\partial(\beta)_i \dd t \\
	    & \ot{(\ref{eq: tec1})}{=}
	    \int_0^1 \int_t^1  \Phi^*_{\tau} [\hc(\beta), \beta_i] \dd\tau \dd t
	    +
	    \int_0^1 \Phi^*_t \hc\partial(\beta)_i \dd t
	  \end{split}
  \end{equation}
  The first integral can be estimated just using  estimate (\ref{eq:Estimate1}) for the
  Chevalley-Eilenberg differential of the difference of two infinitesimal
  momentum maps.
  To estimate the second integral we first need to apply the estimate for
  the homotopy operator (\ref{5.7}) and then again (\ref{eq:Estimate1}).
  Finally, combining these estimates we have
  \begin{equation}
      \| \Phi^*\circ\alpha  - \tilde{\alpha}\|_{k,r}
      \leq
      \| \alpha-\tilde{\alpha} \|_{k+s+2,r(1+\eta)}^2 P(\| \alpha-\tilde{\alpha} \|_{k+s+1,r(1+\eta)})
  \end{equation}
  \end{lemmalist}
\end{proof}

  These estimates finally allow us to prove our main result.
  \begin{theorem}
  \label{thm: riginfPoisson}
  Let us consider a pre-Hamiltonian action of a semisimple compact Poisson Lie group $(G, \pi_G)$
  on a compact Poisson manifold $(M,\pi)$, generated by an infinitesimal
  momentum map $\alpha$.

  There exist a positive integer $l$ and two positive real numbers
  $a$ and $b$ (with $b<1<a$) such that, if $\tilde{\alpha}$ is
  another infinitesimal momentum map on $M$ with respect to the same Poisson
  structure, satisfying
  \begin{equation}
      \| \alpha - \tilde{\alpha} \|_{2l-1}
      \leq
      a
      \qquad
      \| \alpha - \tilde{\alpha} \|_{l} \leq b
  \end{equation}
  then, there exists a Lie algebra morphism $\Phi: \Omega^1(M)\to \Omega^1(M)$
  preserving the chain map property (\ref{eq:Chain})
  of class $\mathscr{C}^{k}$, for all $k\geq l$, on $M$ such that
  $\Phi(\alpha_X) = \tilde\alpha_X$.
  \end{theorem}

  \begin{proof}
    Here we have just to prove that the hypotheses of Theorem \ref{thm:NormalForm} are satisfied.
    First, we make the following indentifications:
    \begin{itemize}
      \item[--] The SCI-space $\mathcal{T}$ is defined to be the space
    $\mathcal{T}_{k}$ of $\mathscr{C}^{k}$-differentiable maps from
    $\g$ to $\Omega(M)$.
      \item[--] The subset $\mathcal{S}$ is given by the infinitesimal momentum maps.
      \item[--] The origin of the affine space is given by $\alpha$ and
    $\mathcal{F} = \mathcal{N} = {0}$ so
    that the estimate (\ref{eqn:proj}) is obvious.
      \item[--] The SCI-group $\mathcal{G}$ consists of the $\mathscr{C}^{k}$-differentiable
    maps from $\Omega^1(M)$ to itself, where the action is
    $\psi\cdot \alpha = \psi\circ \alpha$, with $\psi \in \mathcal{G}$ and
    $\alpha \in \mathcal{T}$.
      \item[--] The closed subgroup $\mathcal{G}_{0}$ of $\mathcal{G}$ is given by the
    Lie algebra morphisms.
    The elements of $\mathcal{G}_{0}$ preserve $\mathcal{S}$.
      \item[--] The SCI-space $\mathcal{H}$ by the space of one-forms on $M$.
    \end{itemize}
    Let us consider the difference of two infinitesimal momentum maps
    $\beta = \alpha-\tilde{\alpha}$  as a 1-cochain  in the Chevalley-Eilenberg
    complex $C^q(\g, \Omega^1(M))$, i.e. an element in $\mathcal{S}$.
    Thus the image of $\beta$ by the map $\mathsf{H}: \mathcal{S}\to \mathcal{H}$
    is simply $\hc_{0}(\beta)$. Using the estimate of the homotopy
    operator Eq.~(\ref{5.7}), the relation (\ref{eqn:estimate-H}) is obvious.
    As stated in Prop.~\ref{prop:FlowProp} the flow $\phi_t$ of the vector field
    $\widehat{X} = \pi^\sharp(\hc(\beta))$ can be recovered by a fiber-wise
    linear automorphism $\Phi_t^*$ on $\Omega^1(M)$, which is
    a Lie algebra morphism (see Eq.~(\ref{eq:FlowLieMorph})).
    By construction, $\Phi_t^*$ commutes with the differential so it preserves
    Eq.~(\ref{eq:Chain}):
    \begin{align*}
        \dd \tilde\alpha_{X}
        &=
        \dd \Phi_t^* \alpha_X \\
        &=
        \Phi_t^* \dd \alpha_X \\
        &\ot{(\ref{eq:Chain}) }{=}
        \Phi_t^*(\alpha \wedge \alpha \circ \delta(X))\\
        &=
        \Phi_t^*(\alpha_{X_i} \wedge \alpha_{X_j}) \\
        &\ot{(\ref{eq:FlowAlgebra})}{ = }
        \Phi_t^*\alpha_{X_i} \wedge \Phi_t^*\alpha_{X_j}\\
        &=
        \tilde\alpha_{X_i} \wedge \tilde\alpha_{X_j}\\
        &=
        \tilde\alpha \wedge \tilde\alpha \circ \delta(X)
    \end{align*}
    Thus, it provides the map
    $\Phi: \mathcal{H} \to \mathcal{G}_{0}$.
    The estimates (\ref{eqn:estimate-Exp})-(\ref{eqn:estimate-Exp-bis})-(\ref{eqn:estimate-zeta})
    are direct consequences of  Lemmas \ref{lem:Estimate1} and \ref{lem:Estimate2}.
  \end{proof}
  Finally we observe that the equivalence of two infinitesimal momentum maps
  implies the equivalence of the corresponding actions under some assumption.
  \begin{corollary}
    Let $\alpha$ and $\tilde\alpha$ two infinitesimal momentum maps
    generating the fundamental vector fields $\widehat{X} = \pi^\sharp(\alpha_X)$
    and $\widehat{X}'= \pi^\sharp(\tilde\alpha)$ of two different actions, resp.
    Under the assumptions of Theorem \ref{thm: riginfPoisson}, if $\alpha_X$ vanishes
    on the symplectic leaves, then $\widehat{X}'= \phi^*\widehat{X}$.
  \end{corollary}
  \begin{proof}
    From Theorem \ref{thm: riginfPoisson} we know that the
    two momentum maps are equivalent, i.e. $\tilde{\alpha}_X = \Phi_t^*\alpha$.
    It is clear that if $\alpha_X$ vanishes on the symplectic leaves,
    we have $\Phi_t^*\alpha_X = \phi_t^* \alpha_X$, where $\phi_t$
    is the flow on $M$ underlying $\Phi_t^*$ (see Theorem \ref{thm:PoissonFlow}).
    Thus,
    \begin{align*}
        \widehat{X}'
        &=
        \pi^\sharp(\tilde\alpha)\\
        &=
        \pi^\sharp (\Phi_t^*\alpha_X)\\
        &=
        \pi^\sharp (\phi_t^*\alpha_X)\\
        &=
        \phi^t_* \pi^\sharp(\alpha_X)\\
        &=
        \phi^t_*\widehat{X}.
    \end{align*}
  \end{proof}
  Since we have used the apparatus of SCI-spaces,
  the analogue of Theorem \ref{thm: riginfPoisson} also holds in the local and semilocal
  case (neighbourhood of an invariant compact submanifold).
  Thus, in the same spirit of \cite{Miranda2012} we also obtain rigidity for pre-Hamiltonian
  actions for actions in a neighbourhood of an invariant compact submanifold
  (which can be reduced to a single point in the case of fixed points for the action).

\bibliographystyle{plain}
\bibliography{references}
\end{document}